\documentclass[11pt,a4paper]{article}
\pdfoutput=1

\usepackage{hyperref}
\usepackage{fullpage}
\usepackage{amssymb}
\usepackage{mathtools}
\usepackage{amsthm}
\usepackage{eucal}
\usepackage{dsfont}
\usepackage[T1]{fontenc}
\usepackage{lmodern}
\usepackage[stretch=10,shrink=10]{microtype}
\usepackage{color,soul}
\usepackage[framemethod=tikz]{mdframed}
\usepackage{lipsum}

\allowdisplaybreaks[1]

\def\F{\mathrn{F}}
\def\X{\CMcal{X}}
\def\Y{\CMcal{Y}}
\def\Z{\CMcal{Z}}

\def\U{\CMcal{U}}

\def\P{\CMcal{P}}

\def\Q{\CMcal{Q}}

\def\O{\CMcal{O}}

\def\H{\CMcal{H}}

\theoremstyle{plain}
\newtheorem{theorem}{Theorem}[section]
\newtheorem{lemma}[theorem]{Lemma}

\theoremstyle{definition}
\newtheorem{definition}[theorem]{Definition}
\newtheorem{claim}[theorem]{Claim}

\newtheorem{remark}[theorem]{Remark}
\newtheorem{fact}[theorem]{Fact}

\newcommand {\br} [1] {\ensuremath{ \left( #1 \right) }}
\newcommand {\Br} [1] {\ensuremath{ \left[ #1 \right] }}

\newcommand {\minusspace} {\: \! \!}
\newcommand {\smallspace} {\: \!}
\newcommand {\fn} [2] {\ensuremath{ #1 \minusspace \br{ #2 } }}
\newcommand {\Fn} [2] {\ensuremath{ #1 \minusspace \Br{ #2 } }}
\newcommand {\defeq} {\ensuremath{ \stackrel{\mathrm{def}}{=} }}
\newcommand {\mutinf} [2] {\fn{\mathrm{I}}{#1 \smallspace : \smallspace #2}}
\newcommand {\condmutinf} [3] {\mutinf{#1}{#2 \smallspace \middle\vert \smallspace #3}}
\newcommand {\prob} [1] {\Fn{\Pr}{#1}}
\newcommand {\abs} [1] {\ensuremath{ \left| #1 \right| }}
\newcommand {\norm} [1] {\ensuremath{ \left\| #1 \right\| }}
\newcommand {\normsub} [2] {\ensuremath{ \norm{#1}_{#2} }}
\newcommand {\onenorm} [1] {\normsub{#1}{1}}
\newcommand {\relent} [2] {\fn{\mathrm{S}}{#1 \middle\| #2}}
\newcommand {\rminent} [2] {\fn{\mathrm{S}_{\infty}}{#1 \middle\| #2}}
\DeclareMathOperator*{\bigE}{\mathbb{E}}
\newcommand {\expec} [2] {\Fn{\bigE_{\substack{#1}}}{#2}}
\newcommand {\email} [1] {\href{mailto:#1}{\texttt{#1}}}

\newcommand {\finset} [1] {\ensuremath{\CMcal{#1}}}
\newcommand {\bra} [1] {\ensuremath{ \left\langle #1 \right| }}
\newcommand {\ket} [1] {\ensuremath{ \left| #1 \right\rangle }}
\newcommand {\ketbratwo} [2] {\ensuremath{ \left| #1 \middle\rangle \middle\langle #2 \right| }}
\newcommand {\ketbra} [1] {\ketbratwo{#1}{#1}}
\newcommand {\cspace} [1] {\ensuremath{\mathnormal{#1}}}

\newcommand {\Tr} {\ensuremath{ \mathrm{Tr} }}
\newcommand {\partrace} [2] {\fn{\Tr_{#1}}{#2}}

\newcommand {\myqedhere} {\tag*{\qedhere}}

\newcommand {\Prot}{\mathcal{P}}
\newcommand {\suppress}[1]{}

\newcommand {\set} [1] {\ensuremath{ \left\lbrace #1 \right\rbrace }}
\def\ve{{\varepsilon}}
\def\F{\mathrm{F}}
\newcommand {\reg} [1] {\ensuremath{ \mathnormal{#1} }}

\def\E{\mathcal{E}}

\newcommand{\qcc}[2]{\mathrm{Q}^{\text{ent},\text{A}\rightarrow\text{B}}_{#1}\br{#2}}
\newcommand{\dqcc}[3]{\mathrm{Q}^{\text{ent},\text{A}\rightarrow\text{B} #1}_{#2}\br{#3}}
\newcommand{\conjugate}[1]{\overline{#1}}
\newcommand{\braket}[2]{\langle#1|#2\rangle}

\newcommand {\mytitle} {New one shot quantum protocols with application to communication complexity}
\newcommand {\Rahul}   {Rahul Jain}
\newcommand {\Anurag}  {Anurag Anshu}
\newcommand {\Penghui} {Penghui Yao}
\newcommand {\Priyanka} {Priyanka Mukhopadhyay}
\newcommand {\Ala}{Ala Shayeghi}

\newcommand {\CQT} {Centre for Quantum Technologies}
\newcommand {\CQTCS} {\CQT{} and Department of Computer Science}
\newcommand {\NUS} {National University of Singapore.}
\newcommand {\IQC} {Institute for Quantum Computing, Unversity of Waterloo}
\newcommand {\Maju} {MajuLab, CNRS-UNS-NUS-NTU International Joint Research Unit, UMI 3654, Singapore.}
\newcommand {\authorblock} [3] {
	\begin{minipage}[t]{0.3\linewidth}
		\centering
		{#1}\\[0.8ex]
		{\footnotesize {#2}\\[-0.7ex]
		\email{#3}}
	\end{minipage}\vspace{1ex}
}

\hypersetup{
	pdfstartview={FitH},
	pdfdisplaydoctitle={true},
	breaklinks={true},
	bookmarksopen={true},
	bookmarksnumbered={false},
	pdftitle={\mytitle},
	pdfauthor={\Anurag, \Rahul, \Priyanka, \Ala, \Penghui}
}

\begin{document}

\begin{titlepage}
\title{\textbf{\mytitle}\\[2ex]}

\author{
    \authorblock{\Anurag}{\CQT, \NUS}{a0109169@u.nus.edu}
	\authorblock{\Rahul}{\CQTCS, \NUS  \\ \Maju}{rahul@comp.nus.edu.sg}
	\authorblock{\Priyanka}{\CQT, \NUS}{a0109168@nus.edu.sg}\\
    \authorblock{\Ala}{\IQC}{ashayeghi@uwaterloo.ca}
	\authorblock{\Penghui}{\CQT, \NUS}{phyao1985@gmail.com}
}

\clearpage\maketitle
\thispagestyle{empty}

\begin{abstract}
In this paper we present the following quantum compression protocol:

\bigskip
\noindent $\Prot$: Let $\rho,\sigma$ be quantum states such that $\relent{\rho}{\sigma} \defeq \Tr (\rho \log \rho - \rho \log \sigma)$, the relative entropy between $\rho$ and $\sigma$, is finite. Alice gets to know the eigen-decomposition of $\rho$. Bob gets to know the eigen-decomposition of $\sigma$. Both Alice and Bob know $ \relent{\rho}{\sigma}$ and an error parameter $\ve$. Alice and Bob use shared entanglement and after communication of $\O((\relent{\rho}{\sigma}+1)/\ve^4)$ bits from Alice to Bob, Bob ends up with a quantum state $\tilde{\rho}$ such that $\F(\rho, \tilde{\rho}) \geq 1 - 5\ve$, where $\F(\cdot)$ represents fidelity.

\bigskip 

This result can be considered as a non-commutative generalization of a result due to Braverman and Rao [2011] where they considered the special case when $\rho$ and $\sigma$ are classical probability distributions (or commute with each other) and use shared randomness instead of shared entanglement. We use $\Prot$ to obtain an alternate proof of a direct-sum result for entanglement assisted quantum one-way communication complexity for all relations, which was first shown by Jain, Radhakrishnan and Sen [2005,2008]. We also present a  variant of protocol $\Prot$ in which Bob has some side information about the state with Alice. We show that in such a case, the amount of communication can be further reduced, based on the side information that Bob has.

Our second result provides a quantum analogue of the widely used classical correlated-sampling protocol. For example, Holenstein [2007] used the classical correlated-sampling protocol in his proof of a parallel-repetition theorem for two-player one-round games. 
\end{abstract}
\end{titlepage}

\section{Introduction}

{\em Relative entropy} is a widely used quantity of central importance in both classical and quantum information theory. In this paper we consider the following task. The notations used below are described in section \ref{sec:Preliminaries}.

\bigskip

\noindent $\Prot$: Given a register $A$, Alice gets to know the eigen-decomposition of a quantum state $\rho \in \mathcal{D}(A)$. Bob gets to know the eigen-decomposition of a quantum state $\sigma\in \mathcal{D}(A)$ such that $\text{supp}(\rho) \subset \text{supp}(\sigma)$. Both Alice and Bob know $ \relent{\rho}{\sigma} \defeq \Tr \rho \log \rho - \rho \log \sigma$, the relative entropy between $\rho$ and $\sigma$ and an error parameter $\ve$. Alice and Bob use shared entanglement and after communication of $\O((\relent{\rho}{\sigma}+1)/\ve^4)$ bits from Alice to Bob, Bob ends up with a quantum state $\tilde{\rho}$ such that $\F(\rho, \tilde{\rho}) \geq 1 - \ve$, where $\F(\cdot,\cdot)$ represents {\em fidelity}.
%\footnote{We can also consider a  more stringent condition that $\relent{\rho}{\sigma}$ is not known to Alice and Bob. We state here, without any proof, that this case can also be handled with  communication $\O(\relent{\rho}{\sigma})$ however now using multiple rounds of communication, along similar lines as done by Braverman and Rao~\cite{Braverman2011} for the special case when the inputs to Alice and Bob are probability distributions.}

\bigskip 

This result can be considered as a non-commutative generalization of a result due to Braverman and Rao~\cite{Braverman2011} where they considered the special case when $\rho$ and $\sigma$ are classical probability distributions and the two parties only share public random coins. Their protocol, and slightly modified versions of it, were widely used to show several {\em direct sum} and {\em direct product} results in communication complexity, for example a direct sum theorem for all relations in the bounded-round public-coin communication model~\cite{Braverman2011}, direct product theorems for all relations in the public-coin one-way and public-coin bounded-round communication models~\cite{Jain:2011,Jain:2012,BravermanRWY:2013}.  A direct sum result for a relation $f$ in a model of communication  (roughly) states that in order to compute $k$ independent  instances of $f$ simultaneously, if we provide communication less than $k$ times the communication required to compute $f$ with the constant success probability $p<1$, then the success probability for computing all the $k$ instances of $f$ correctly is at most a constant $q<1$. A direct product result, which is a stronger result, states that in such a situation  the success probability for computing all the $k$ instances of $f$ correctly is at most $p^{-\Omega(k)}$.

Protocol $\Prot$ allows for compressing the communication in one-way entanglement-assisted quantum communication protocols to the {\em internal information} about the inputs carried by the message. Using this we obtain a direct-sum result for {\em distributional entanglement assisted quantum one-way communication complexity} for all relations. This direct-sum result was shown previously by Jain, Radhakrishnan and Sen~\cite{Jain:2005:PEM:1068502.1068658,Jain:2008} and they obtained this result via a protocol that allowed them compression to {\em external information} carried in the message\footnote{Compression to external and internal information can be thought of as one-shot communication analogues of the celebrated results by Shannon~\cite{Shannon:1948} and Slepian-Wolf~\cite{Slepian:1973} exhibiting compression of source to entropy and conditional entropy respectively.}. Their arguments are quite specific to one-way protocols and  do not seem to generalize to multi-round communication protocols. Our proof however, is along the lines of a proof which has been generalized to bounded-round classical protocols~\cite{Braverman2011} and hence it presents hope that our direct-sum result can also be generalized to  bounded-round quantum protocols. The protocol of Braverman and Rao~\cite{Braverman2011} was also used by Jain~\cite{Jain:2011} to obtain a direct-product for all relations in the model of one-way public-coin classical communication  and later extended to multiple round  public-coin classical communication~\cite{Jain:2012,BravermanRWY:2013}. Hence  protocol $\Prot$ also presents a hope of obtaining similar results for quantum communication protocols.

We also present a variant of protocol $\Prot$, with Bob possessing some side information about Alice's input. In such a case, the communication can be further reduced. 

\bigskip 

\noindent $\Prot'$: Given two registers $A$ and $B$, Alice and Bob know the description of a quantum channel $\E: \mathcal{L}(A)\rightarrow \mathcal{L}(B)$. Alice is given the eigen-decomposition of a state $\rho \in \mathcal{D}(A)$. Bob is given the eigen-decomposition of a state $\sigma \in A$ (such that $\text{supp}(\rho) \subset \text{supp}(\sigma)$) and the state $\rho'= \E(\rho)$. Let  $\relent{\rho}{\sigma}-\relent{\E(\rho)}{\E(\sigma)}$ and $\ve >0$ be known to Alice and Bob. There exists a protocol, in which Alice and Bob use shared entanglement and Alice sends $\O((\relent{\rho}{\sigma}-\relent{\E(\rho)}{\E(\sigma)}+1)/\ve^4)$ bits of communication to Bob, such that with probability at least $1-4\ve$, the state $\tilde{\rho}$ that Bob gets at the end of the protocol satisfies  $\F(\rho,\tilde{\rho})\geq 1-\ve$, where $\F(\cdot,\cdot)$ represents {\em fidelity} .

\bigskip 

In the second part of our paper, we present the following protocol, which can be considered as a quantum analogue of the widely used {\em classical correlated sampling} protocol. For example, Holenstein~\cite{Holenstein2007} has used the classical correlated sampling protocol in his proof of a {\em parallel-repetition theorem} for two-player one-round games.

\bigskip 

\noindent $\Prot_1:$ Given a register $A_1$, Alice gets to know the eigen-decomposition of a quantum state $\rho\in \mathcal{D}(A_1)$. Bob gets to know the eigen-decomposition of a quantum state $\sigma \in \mathcal{D}(A_1)$. Alice and Bob use shared entanglement, do local measurements (no communication) and at the end Alice outputs registers $A_1A_2$ and Bob outputs registers $B_1B_2$ such that the following holds:
\begin{enumerate}
\item $B_1\equiv A_1$ and $B_2\equiv A_2$.
\item The marginal state in register $A_1$ is $\rho$ and the marginal state in register $B_1$  is $\sigma$.
\item For any projective measurement $M  = \{M_1, \ldots , M_w\}$ such that $M_i \in \mathcal{L}(A_1A_2)$, the following holds. Let Alice perform $M$ on $A_1A_2$ and Bob perform $M$ on $B_1B_2$  and obtain outcomes $I \in [w], J \in [w]$ respectively. Then,
$$ \Pr[I = J]  \geq  \left(1-\sqrt{\onenorm{\rho-\sigma} - \frac{1}{4}\onenorm{\rho-\sigma}^2}\right)^{3} . $$
\end{enumerate}

Recently, Dinur, Steurer and Vidick~\cite{DinurSteurerVidick:2013} have shown another version of a  quantum correlated sampling protocol different from ours, and used it in their proof of a parallel-repetition theorem for two-prover one-round entangled projection games.

\subsection*{Our techniques}

Our protocol $\Prot$ is inspired by the protocol of Braverman and Rao~\cite{Braverman2011}, which as we mentioned, applies to the special case when inputs to Alice and Bob are classical probability distributions $P,Q$ respectively. Let us first assume the case when Alice and Bob know $c = \rminent{P}{Q} \defeq \min\{\lambda |~ P \leq 2^\lambda Q \}$, the {\em relative max-entropy} between $P$ and $Q$. In the protocol of~\cite{Braverman2011}, Alice and Bob share (as public coins) $\{(M_i,R_i)|~ i \in \mathbb{N}\}$, where each $(M_i,R_i)$ is independently and identically distributed uniformly over $\U\times[0,1]$, $\U$ being the support of $P$ and $Q$. Alice accepts index $i$ iff $ R_i \leq P(M_i) $ and Bob accepts index $i$ iff $ R_i \leq  2^c Q(M_i) $. It is easily argued that for the first index $j$ accepted by Alice, $M_j$ is distributed according to $P$. Braverman and Rao argue that Alice can communicate this index $j$ to Bob, with high probability, using communication of $\O(c)$ bits (for constant $\ve$), using crucially the fact that $P \leq 2^c Q $.  

In our protocol, Alice and Bob share infinite copies of the following quantum state 
\[\ket{\psi} \defeq \frac{1}{\sqrt{NK}}\sum_{i=1}^N\ket{i}^A\ket{i}^B\otimes \left(\sum_{m=1}^K\ket{m}^{A_1}\ket{m}^{B_1} \right),\]
where registers $A,B$ serve to sample a maximally mixed state in the support of $\rho, \sigma$ and the registers $A_1,B_1$ serve to sample uniform distribution in the interval $[0,1]$ (in the limit $K \rightarrow \infty$).  Again let us first assume the case when Alice and Bob know  $c = \rminent{\rho}{\sigma} \defeq \min\{\lambda |~ \rho \leq 2^\lambda \sigma \}$ (here $\leq$ represents the L\"{o}wner order), the relative max-entropy between $\rho$ and $\sigma$.  Let eigen-decomposition of $\rho$ be $ \sum_{i=1}^N a_i \ket{a_i}\bra{a_i}$ and eigen-decomposition of $\sigma$ be $\sum_{i=1}^N b_i \ket{b_i}\bra{b_i}$. Consider a projection $P_{AA_1}$ as defined below and $I_{AA_1}$ the identity operator on registers $A,A_1$. Alice performs a measurement $\{P_{AA_1},I_{AA_1}-P_{AA_1}\}$, on the register $AA_1$ of each copy of $\ket{\psi}$ and \textit{accepts} the index of a copy iff outcome of measurement corresponds to $P_{AA_1}$ (which we refer to as a success for Alice). 
\begin{equation*}
P_{AA_1}=\sum_{i=1}^N\ket{a_i}_A\bra{a_i}_A\otimes \left(\sum_{m=1}^{\lceil Ka_i \rceil}\ket{m}_{A_1}\bra{m}_{A_1}\right) .
\end{equation*}
Similarly, consider a projection $P_{BB_1}$ as defined below (for an appropriately chosen $\delta$) and $I_{BB_1}$ the identity operator on register $BB_1$. Bob performs a measurement $\{P_{BB_1},I_{BB_1}-P_{BB_1}\}$  on registers $BB_1$ on each copy of $\ket{\psi}$ and \textit{accepts} the index of a copy iff the outcome of measurement corresponds to $P_{BB_1}$ (which we refer to as a success for Bob).
\begin{equation*}
P_{BB_1}=\sum_{i=1}^N\ket{b_i}_B\bra{b_i}_B\otimes \left(\sum_{m=1}^{\text{min}\{\lceil 2^{c}K b_i/\delta\rceil, K\}}\ket{m}_{B_1}\bra{m}_{B_1}\right) .
\end{equation*}
Again it is easily argued that (in the limit $K \rightarrow \infty$) the marginal state in $B$ (and also in $A$), in the first copy of $\ket{\psi}$ on which Alice succeeds, is $\rho$.  Using crucially the fact that $\rho \leq 2^c \sigma$, we argue that after Alice's measurement succeeds in a copy, Bob's measurement also succeeds with high probability. Hence, by {\em gentle measurement lemma} (\cite{Winter:1999,Ogawa:2002}), the marginal state in register $B$ is not disturbed much, conditioned on success of both Alice and Bob.  We also argue that Alice can communicate the index of this copy to Bob with communication of $\O(c)$ bits (for constant $\ve$).

As can be seen, our protocol is a natural quantum analogue of the protocol of Braverman and Rao~\cite{Braverman2011}.  However, since $\rho$ and $\sigma$ may not commute, our analysis deviates significantly from the analysis of~\cite{Braverman2011}. We are required to show several new facts related to the non-commuting case while arguing that the protocol still works correctly.  

We then consider the case in which  $\relent{\rho}{\sigma}$ (instead of $\rminent{\rho}{\sigma}$) is known to Alice and Bob. The {\em quantum substate theorem}~\cite{Jain2002, JainNayak:2012} implies that there exists a quantum state $\rho'$, having high fidelity with $\rho$ such that  $\rminent{\rho'}{\sigma} = \O(\relent{\rho}{\sigma})$. We argue that our protocol is robust with respect to  small perturbations in Alice's input and hence works well for the pair $(\rho', \sigma)$ as well, and uses  communication $\O(\relent{\rho}{\sigma})$ bits. Again this requires us to show new facts related to the non-commuting case.

\subsection*{Related work} 
 Much progress has been made in the last decade towards proving direct sum and direct product conjectures in various models of communication complexity and information theory has played a crucial role in these works. Most of the proofs have build upon elegant one-shot  protocols for interesting information theoretic tasks. For example, consider the following task which is a special case of the task we consider in the protocol $\Prot$.

\bigskip

\noindent {\bf T1}:  Alice gets to know the eigen-decomposition of a quantum state $\rho$. Alice and Bob get to know the eigen-decomposition of a quantum state $\sigma$, such that $\text{supp}(\rho) \subset \text{supp}(\sigma)$. They also know $c \defeq \relent{\rho}{\sigma}$, the relative entropy between $\rho$ and $\sigma$ and an error parameter $\ve$. They  use shared entanglement and communication and at the end of the protocol, Bob ends up with a quantum state $\tilde{\rho}$ such that $\F(\rho, \tilde{\rho}) \geq 1 - \ve$.

\bigskip

Jain, Radhakrishnan and Sen in~\cite{Jain:2005:PEM:1068502.1068658,Jain:2008}, showed that this task  (for constant $\ve$) can be achieved with  communication $\O(\relent{\rho}{\sigma} +1)$ bits, and this led to direct sum theorems for all relations in  entanglement-assisted quantum one-way and entanglement-assisted quantum simultaneous message-passing communication models. They also considered the special case when the inputs to Alice and Bob are probability distributions $P,Q$ respectively and showed that sharing public random coins and $\O(\relent{P}{Q}+1))$ bits of  communication can achieve this task (for constant $\ve$). Later an improved result was obtained by Harsha, Jain, Mc. Allester and Radhakrishnan~\cite{Harsha:2010}, where they presented a protocol in which Bob is able to sample exactly from $P$ with expected communication $\relent{P}{Q}+2\log\relent{P}{Q}+\O(1)$.  This led to direct sum theorems for all relations in the public-coin  randomized one-way, public-coin simultaneous message passing~\cite{Jain:2005:PEM:1068502.1068658,Jain:2008} and public-coin  randomized bounded-round communication models~\cite{Harsha:2010}.

Our work strengthens their results by showing that $O(\relent{\rho}{\sigma})$ bits of communication is enough even if $\sigma$ is not known to Alice.

Very recently, Touchette~\cite{Touchette:2015} introduced the notion of {\em quantum information cost} which generalizes the internal information cost in the classical communication to the quantum setting.  Moreover, he showed that in {\em bounded-round entanglement assisted quantum communication tasks}, the communication can be compressed to the quantum information cost based on the {\em state redistribution} protocol~\cite{DevetakY:2008,YardD:2009}. Using such a compression protocol, he showed a direct sum theorem for bounded round entanglement assisted quantum communication model. 

\subsection*{Organization} In section \ref{sec:Preliminaries}, we discuss our notations and relevant notions needed for our proofs. In Section~\ref{sec:main} we describe our one shot quantum protocol $\P$.  The direct sum result follows in Section~\ref{sec:directsum}. In Section~\ref{sec:qcorr} we present quantum correlated sampling. We conclude in Section ~\ref{sec:conclusion}

\section{Preliminaries}
\label{sec:Preliminaries}

In this section we present some notations, definitions, facts and lemmas that we will use later in our proofs.

\subsection*{Information theory}

For integer $n \geq 1$, let $[n]$ represent the set $\{1,2, \ldots, n\}$.  We let $\log$ represent logarithm to the base $2$ and $\ln$ represent logarithm to the base $\mathrm{e}$. Let $\finset{X}$ and $\finset{Y}$ be finite sets. $\finset{X}\times\finset{Y}$ represents the cross product of $\finset{X}$ and $\finset{Y}$. For a natural number $k$, we let 
$\finset{X}^k$ denote the set $\finset{X}\times\cdots\times\finset{X}$, the cross product of
$\finset{X}$, $k$ times.  Let $\mu$ be a probability distribution on $\finset{X}$. We let $\mu(x)$ represent the probability of $x\in\finset{X}$ according to $\mu$. We use the same symbol to represent a random variable and its distribution whenever it is clear from the context. The expectation value of function $f$ on $\finset{X}$ is defined as
$\expec{x \leftarrow X}{f(x)} \defeq \sum_{x \in \finset{X}} \prob{X=x}
\cdot f(x),$ where $x\leftarrow X$ means that $x$ is drawn according to distribution $X$.

Consider a Hilbert space $\H$ endowed with an inner product $\langle \cdot, \cdot \rangle$. The $\ell_1$ norm of an operator $X$ on $\H$ is $\onenorm{X}\defeq\Tr\sqrt{X^{\dag}X}$ and $\ell_2$ norm is $\norm{X}_2\defeq\sqrt{\Tr XX^{\dag}}$. A quantum state (or a density matrix or just a state) is a positive semi-definite matrix with trace equal to $1$. It is called {\em pure} if and only if the rank is $1$. A sub-normalized state is a positive semi-definite matrix with trace less than or equal to $1$. Let $\ket{\psi}$ be a unit vector on $\H$, that is $\langle \psi,\psi \rangle=1$.  With some abuse of notation, we use $\psi$ to represent the state and also the density matrix $\ketbra{\psi}$, associated with $\ket{\psi}$. 

Fix an orthonormal basis on $\H$, referred to as {\em computational basis}. Let $\conjugate{\ket{\psi}}$ represent the complex conjugation of $\ket{\psi}$, taken in the computational basis. A classical distribution $\mu$ can be viewed as a quantum state with non-diagonal entries $0$.  Given a quantum state $\rho$ on $\H$, {\em support of $\rho$}, called $\text{supp}(\rho)$ is the subspace of $\H$ spanned by all eigen-vectors of $\rho$ with non-zero eigenvalues.
 
A {\em quantum register} $A$ is associated with some Hilbert space $\H_A$. Define $|A| \defeq \dim(\H_A)$. Let $\mathcal{L}(A)$ represent the set of all linear operators on $\H_A$. We denote by $\mathcal{D}(A)$, the set of quantum states on the Hilbert space $\H_A$. State $\rho$ with subscript $A$ indicates $\rho_A \in \mathcal{D}(A)$. If two registers $A,B$ are associated with the same Hilbert space, we shall represent the relation by $A\equiv B$.  Composition of two registers $A$ and $B$, denoted $AB$, is associated with Hilbert space $\H_A \otimes \H_B$.  For two quantum states $\rho\in \mathcal{D}(A)$ and $\sigma\in \mathcal{D}(B)$, $\rho\otimes\sigma \in \mathcal{D}(AB)$ represents the tensor product (Kronecker product) of $\rho$ and $\sigma$. The identity operator on $\H_A$ (and associated register $A$) is denoted $I_A$. 

Let $\rho_{AB} \in \mathcal{D}(AB)$. We define
\[ \rho_{\reg{B}} \defeq \partrace{\reg{A}}{\rho_{AB}}
\defeq \sum_i (\bra{i} \otimes I_{\cspace{B}})
\rho_{AB} (\ket{i} \otimes I_{\cspace{B}}) , \]
where $\set{\ket{i}}_i$ is an orthonormal basis for the Hilbert space $\H_A$.
The state $\rho_B\in \mathcal{D}(B)$ is referred to as the marginal state of $\rho_{AB}$. Unless otherwise stated, a missing register from subscript in a state will represent partial trace over that register. Given a $\rho_A\in\mathcal{D}(A)$, a {\em purification} of $\rho_A$ is a pure state $\rho_{AB}\in \mathcal{D}(AB)$ such that $\partrace{\reg{B}}{\rho_{AB}}=\rho_A$. A purification of a quantum state is not unique.

A quantum {map} $\E: \mathcal{L}(A)\rightarrow \mathcal{L}(B)$ is a completely positive and trace preserving (CPTP) linear map (mapping states in $\mathcal{D}(A)$ to states in $\mathcal{D}(B)$). A {\em unitary} operator $U_A:\H_A \rightarrow \H_A$ is such that $U_A^{\dagger}U_A = U_A U_A^{\dagger} = I_A$. An {\em isometry}  $V:\H_A \rightarrow \H_B$ is such that $V^{\dagger}V = I_A$ and $VV^{\dagger} = I_B$. The set of all unitary operations on register $A$ is  denoted by $\mathcal{U}(A)$.
%Readers can refer to~\cite{CoverT91,NielsenC00,Watrouslecturenote} for more details.

\begin{definition}
We shall consider the following information theoretic quantities. Let $A$ be a quantum register. Let $\varepsilon \geq 0$. 
\begin{enumerate}
\item {\bf Fidelity} For $\rho,\sigma \in \mathcal{D}(A)$, $$\F(\rho,\sigma)\defeq\onenorm{\sqrt{\rho}\sqrt{\sigma}}.$$ For classical probability distributions $P = \{p_i\}, Q =\{q_i\}$, $$\F(P,Q)\defeq \sum_i \sqrt{p_i \cdot q_i}.$$
%\item {\bf Trace distance} For $\rho,\sigma \in \mathcal{D}(A)$, trace distance is given by $\frac{1}{2}\onenorm{\rho-\sigma}$.
\item {\bf Entropy} For $\rho\in\mathcal{D}(A)$, $$S(\rho_A) \defeq - \Tr(\rho_A\log\rho_A) .$$ 
\item {\bf Relative entropy} For $\rho,\sigma\in \mathcal{D}(A)$ such that $\text{supp}(\rho) \subset \text{supp}(\sigma)$, $$\relent{\rho}{\sigma} \defeq \Tr(\rho\log\rho) - \Tr(\rho\log\sigma) .$$ 
\item {\bf Relative max-entropy} For $\rho,\sigma\in \mathcal{D}(A)$ such that $\text{supp}(\rho) \subset \text{supp}(\sigma)$, $$ \rminent{\rho}{\sigma}  \defeq  \inf \{ \lambda \in \mathbb{R} : 2^{\lambda} \sigma \geq \rho \}  .$$
\item {\bf Mutual information} For $\rho_{AB}\in \mathcal{D}(AB)$, $$\mutinf{A}{B}_{\rho}\defeq S(\rho_A) + S(\rho_B)-S(\rho_{AB}) = \relent{\rho_{AB}}{\rho_A\otimes\rho_B}.$$
\item {\bf Conditional mutual information} For $\rho_{ABC}\in\mathcal{D}(ABC)$, $$\condmutinf{A}{B}{C}_{\rho}\defeq \mutinf{A}{BC}_{\rho}-\mutinf{A}{C}_{\rho}.$$
\end{enumerate}
\label{def:infquant}
\end{definition}

We will use the following facts.
\begin{fact}[\cite{NielsenC00} page 416]
	\label{fact:tracefidelityequi}
	For quantum states $\rho,\sigma \in \mathcal{D}(A)$, it holds that
	\[2(1-\F(\rho,\sigma))\leq\onenorm{\rho-\sigma}\leq2\sqrt{1-\F(\rho,\sigma)^2}.\]
For two pure states $\ket{\phi}$ and $\ket{\psi}$, we have
\[\onenorm{\phi-\psi}=2 \sqrt{1-\F(\phi, \psi)^2}=2 \sqrt{1-\abs{\langle\phi|\psi\rangle}^2}.\]
\end{fact}

\begin{fact}[\cite{stinespring55}](\textbf{Stinespring representation})\label{stinespring}
Let $\E(\cdot): \mathcal{L}(A)\rightarrow \mathcal{L}(B)$ be a quantum operation. There exists a Hilbert space $C$ and an unitary $U: A \otimes B \otimes C \rightarrow A \otimes B\otimes C$ such that $\E(\omega)=\Tr_{A,C}\br{U (\omega  \otimes \ketbra{0}^{B,C}) U^{\dagger}}$. Stinespring representation for a channel is not unique. 
\end{fact}

\begin{fact}[\cite{barnum96},\cite{lindblad75}]
	\label{fact:monotonequantumoperation}
For states $\rho$, $\sigma \in \mathcal{D}(A)$, and quantum operation $\E(\cdot):\mathcal{L}(A)\rightarrow \mathcal{L}(B)$, it holds that
\begin{align*}
	\onenorm{\E(\rho) - \E(\sigma)} \leq \onenorm{\rho - \sigma} \quad \mbox{and} \quad \F(\E(\rho),\E(\sigma)) \geq \F(\rho,\sigma) \quad \mbox{and} \quad \relent{\rho}{\sigma}\geq \relent{\E(\rho)}{\E(\sigma)}.
\end{align*}
In particular, for bipartite states $\rho^{AB},\sigma^{AB}\in \mathcal{D}(AB)$, it holds that
\begin{align*}
	\onenorm{\rho^{AB} - \sigma^{AB}} \geq \onenorm{\rho^A - \sigma^A} \quad \mbox{and} \quad \F(\rho^{AB},\sigma^{AB}) \leq \F(\rho^A,\sigma^A) \quad \mbox{and} \quad \relent{\rho_{AB}}{\sigma_{AB}}\geq \relent{\rho_A}{\sigma_A} .
\end{align*}
\end{fact}

\begin{fact}[\cite{Watrouslecturenote} Lemma 4.41.]\label{fact:onenorm vs twonorm}
Let $A, B$ be two positive semidefinite operators on Hilbert space $\H$. Then
\[\onenorm{A-B}\geq\norm{\sqrt{A}-\sqrt{B}}^2_2.\]
\end{fact}
\begin{fact}\label{fact:fidelityvstrace}
Given two quantum states $\rho$ and $\sigma$, 
$$\Tr\sqrt{\rho}\sqrt{\sigma}\geq 1 - \frac{1}{2} \onenorm{ \rho - \sigma} \geq 1-\sqrt{1-\F(\rho,\sigma)^2} .$$
\end{fact}
\begin{proof}
By  Facts~\ref{fact:onenorm vs twonorm} and~\ref{fact:tracefidelityequi},
\[2\sqrt{1-\F(\rho,\sigma)^2}\geq \onenorm{ \rho - \sigma } \geq\norm{\sqrt{\rho}-\sqrt{\sigma}}_2^2=2-2 \cdot \Tr\br{\sqrt{\rho}\sqrt{\sigma}}. \myqedhere\]
\end{proof}

% For a unitary $U: \rminent{U\rho U^{\dagger}}{U\sigma U^{\dagger}}=\rminent{\rho}{\sigma}$ and $\relent{U\rho U^{\dagger}}{U\sigma U^{\dagger}}=\relent{\rho}{\sigma}$.
%Since  logarithm is operator-monotone, we have $\relent{\rho}{\sigma} \leq \rminent{\rho}{\sigma}$.

\begin{fact}[Joint concavity of fidelity][\cite{Watrouslecturenote}, Proposition 4.7]\label{fact:fidconcave}
Given states $\rho_1,\rho_2\ldots\rho_k,\sigma_1,\sigma_2\ldots\sigma_k$ and positive numbers $p_1,p_2\ldots p_k$ such that $\sum_ip_i=1$. Then $$\F(\sum_ip_i\rho_i,\sum_ip_i\sigma_i)\geq \sum_ip_i\F(\rho_i,\sigma_i).$$
\end{fact}

\begin{fact}[\cite{JainRS09,JainNayak:2012}](\textbf{Quantum substrate theorem})\label{fact:substrate}
Given $\rho, \sigma \in \mathcal{D}(A)$, such that $\text{supp}(\rho) \subset \text{supp}(\sigma)$. For any $\ve>0$, there exists $\rho'\in \mathcal{D}(A)$ such that
\[\F(\rho,\rho')\geq1-\ve\quad \mbox{and}\quad \rminent{\rho'}{\sigma}\leq\frac{\relent{\rho}{\sigma}+1}{\ve}+\log\frac{1}{1-\ve}.\]
\end{fact}

\begin{fact}[\cite{Winter:1999,Ogawa:2002}](\textbf{Gentle measurement lemma})\label{fact:gentlemeasurement}
Let $\rho\in \mathcal{D}(A)$  and $\Pi$ be a projector. Then,
\[\F(\rho,\frac{\Pi \rho \Pi}{\Tr \Pi \rho })\geq \sqrt{\Tr \Pi \rho}.\]
\end{fact}
\begin{proof}
Introduce a register $B$, such that $|B|\geq |A|$. Let $\phi\in \mathcal{D}(AB)$ be a purification of $\rho$. Then $( \Pi \otimes I_B)\phi ( \Pi \otimes I_B)$ is a purification of $\Pi \rho \Pi $. Hence  (using monotonicity of fidelity under quantum operation, Fact~\ref{fact:monotonequantumoperation})
\[\F(\rho, \frac{\Pi \rho \Pi}{\Tr \Pi \rho } )\F\br{\phi,(\Pi\otimes I_B)\phi(\Pi\otimes I_B)}= \frac{\abs{\bra{\phi}(\Pi \otimes I) \ket{\phi}}}{\norm{(\Pi \otimes I) \ket{\phi}}}  = \sqrt{\Tr( \Pi \rho)}.   \myqedhere \] 
\end{proof}

\begin{fact}
\label{fact:ruskai}
Given quantum states $\sigma_{AB}\in \mathcal{D}(AB), \rho_A\in \mathcal{D}(A)$, such that $\text{supp}(\rho_A)\subset\text{supp}(\sigma_A)$, it holds that 
$$\Tr(e^{\text{log}(\sigma_{AB})-\text{log}(\sigma_{A}\otimes I_B)+\text{log}(\rho_{A}\otimes I_B)})<1.$$
\end{fact}
\begin{proof}
Consider, 
\begin{eqnarray*}
\Tr(e^{\text{log}(\sigma_{AB})-\text{log}(\sigma_{A}\otimes I_B)+\text{log}(\rho_{A}\otimes I_B)}) &<& \int_{0}^{\infty}du\Tr(\sigma_{AB}\frac{1}{\sigma_A+uI_A}\rho_A\frac{1}{\sigma_A+uI_A}) \quad (\text{Theorem} 5, ~\cite{ruskai2002} ) \\ 
&=&  \int_{0}^{\infty}du\Tr(\frac{1}{\sigma_A+uI_A}\sigma_{A}\frac{1}{\sigma_A+uI_A}\rho_A) \\
&=& \Tr(\sigma_{A}\int_{0}^{\infty}du\frac{1}{(\sigma_A+uI_A)^2}\rho_A) = \Tr(\sigma_A\sigma^{-1}_A\rho_A) =1.
\end{eqnarray*}
\end{proof}

\begin{fact}~\cite{Lieb:1973,LiebR:1973}(\textbf{Strong subadditivity theorem})\label{fact:strongsubadditivity}
For any tripartite quantum state $\rho\in\mathcal{D}(ABC)$, it holds that $\condmutinf{A}{C}{B}_{\rho}\geq0$.	
\end{fact}

\begin{fact}[\cite{LiebAraki:1970} and \cite{NielsenC00}, page 515]\label{fact:conditionalmutinf}
For a quantum state $\rho_{AB}\in\mathcal{D}(AB)$, it holds that $\abs{\mathrm{S}(\rho_A)-\mathrm{S}(\rho_B)}\leq\mathrm{S}(\rho_{AB})\leq\mathrm{S}(\rho_A)+\mathrm{S}(\rho_B)$. Furthermore,
\[\mutinf{A}{B}_{\rho}=\mathrm{S}(\rho_A)+\mathrm{S}(\rho_B)-\mathrm{S}(\rho_{AB})\leq 2\mathrm{S}(\rho_A).\]
\end{fact}

\begin{fact}\label{fact:chainrulemutualinf}
Let $\rho_{A_1A_2\ldots A_kBC} \in \mathcal{D}(A_1\cdots A_k BC)$ such that $\rho_{A_1A_2\ldots A_k}=\rho_{A_1}\otimes\rho_{A_2}\otimes\ldots\rho_{A_k}$. Then,
\[\condmutinf{A_1A_2\ldots A_k}{B}{C}_{\rho}\geq\sum_{i=1}^k\condmutinf{A_i}{B}{C}_{\rho}.\]
\end{fact}
\begin{proof}
Consider,
\begin{eqnarray*}
\condmutinf{A_1A_2\ldots A_k}{B}{C}_{\rho}&=&\condmutinf{A_1}{B}{C}_{\rho}+\condmutinf{A_2A_3\ldots A_k}{B}{A_1C}_\rho\\
&=&\condmutinf{A_1}{B}{C}_{\rho}+\mutinf{A_2A_3\ldots A_k}{A_1BC}_{\rho}-\mutinf{A_1}{A_2A_3\ldots A_k}\rho\\
&=&\condmutinf{A_1}{B}{C}_{\rho}+\mutinf{A_2A_3\ldots A_k}{A_1BC}_{\rho}\\
&\geq&\condmutinf{A_1}{B}{C}_{\rho}+\condmutinf{A_2A_3\ldots A_k}{B}{C}_{\rho}
\end{eqnarray*}
The first and second equalities follow from the definition of the conditional mutual information. The third equality is from the independence between $A_1$ and $A_2A_3\ldots A_k$. The last inequality is from strong subadditivity (Fact~\ref{fact:strongsubadditivity}). Proof follows by induction.
\end{proof}

For the facts appearing below, the proofs can be obtained by direct calculations and hence have been skipped. 

\begin{fact}
	\label{fact:relative entropy splitting}
	Given $\rho_{AB},\sigma_{AB}\in\mathcal{D}(AB)$, such that $\text{supp}(\sigma_{AB})\subset \text{supp}(\rho_{AB})$, $\rho_{AB} = \sum_a \mu(a) \ketbra{a}_A \otimes\rho^a_B$ and
	$\sigma_{AB} = \sum_a \mu'(a) \ketbra{a}_A \otimes\sigma^a_B$, where $\rho^a_B,\sigma^a_B\in \mathcal{D}(B)$, $\mu(a),\mu'(a)\geq 0$ and $\sum_a\mu(a)=1,\sum_a\mu'(a)=1$.
	It holds from the definition of relative entropy that
	\[ \relent{\sigma_{AB}}{\rho_{AB}} = \relent{\mu}{\mu'}
	+ \expec{a\leftarrow \mu'} {\relent{\sigma^a_B}{\rho^a_B}}.\]
\end{fact}

\begin{fact}
Given a classical-quantum state $\rho_{AB}\in \mathcal{D}(AB)$ of the form
$\rho_{AB} = \sum_a \mu(a) \ketbra{a}_A \otimes \rho^a_B$,
where $\rho^a_B\in\mathcal{D}(B)$ and $\sum_a \mu(a)=1$, $\mu(a)\geq 0$, we have
\[ \mutinf{A}{B}_{\rho}
= \mathrm{S}\br{\sum_a\mu(a)\rho_a}-\sum_a\mu(a)\mathrm{S}\br{\rho_a},\]
\end{fact}

\begin{fact}
\label{fact:cqcondmut}
Let $\rho_{ABC}$ be a state of the form $\rho_{ABC}=\sum_c \mu(c)\ketbra{c}_{C}\otimes \rho^c_{AB}$, where $\rho^c_{AB}\in\mathcal{D}(AB)$ and $\sum_c\mu(c)=1$, $\mu(c)\geq 0$. Then
$$\condmutinf{A}{B}{C}_{\rho}=\sum_c \mu(c)\mutinf{A}{B}_{\rho^c}.$$
\end{fact}

\subsection*{Communication complexity}
In this section we briefly describe entanglement assisted quantum one-way communication complexity. A mathematically detailed definition has been given by Touchette in \cite{Touchette:2014}. Let $f\subseteq\X\times\Y\times\Z$ be a relation. Alice holds input $x \in \X$ and Bob holds input $y \in \Y$. They may share prior quantum states independent of the inputs. Alice makes a unitary transformation on her qubits, based on her input $x$, and sends part of her qubits to Bob. Bob makes a unitary operation, based on his input $y$, and measures the last few qubits (answer registers) in the computational basis to get the answer $z \in \Z$. The answer is declared correct if $(x,y,z)\in f$. Let $\qcc{\ve}{f}$ represent the quantum one-way communication complexity of $f$ with worst case error $\ve$, that is minimum number of qubits Alice needs to send to Bob, over all protocols computing $f$ with error at most $\ve$ on \textit{any} input $(x,y)$.

We let $\dqcc{, \mu}{\ve}{f}$ represent distributional quantum one-way communication complexity of $f$ under distribution $\mu$ over $\X\times \Y$ with distributional error at most $\ve$. This is the communication cost of the best protocol computing $f$ with maximum error $\ve$ averaged over distribution $\mu$. Following is Yao's min-max theorem connecting the worst case error and the distributional error settings.

\begin{fact}\label{fact:yaos principle}\cite{Yao:1979:CQR:800135.804414}
$\qcc{\ve}{f}=\max_{\mu}\dqcc{, \mu}{\ve}{f}$.
\end{fact}

\section{A quantum compression protocol}
\label{sec:main}
Following is our main result in this section.

\begin{theorem}\label{thm:protocol}
\label{qbev}
Given quantum states $\rho,\sigma$ on a Hilbert space $\H$ with dimension $N$, such that $\text{supp}(\rho) \subset \text{supp}(\sigma)$. Alice is given the eigen-decomposition of $\rho$ and Bob is given the eigen-decomposition of $\sigma$. Let  $\relent{\rho}{\sigma}$ and $\ve >0$ be known to Alice and Bob. There exists an entanglement assisted quantum one-way communication protocol, with Alice sending $\O(\relent{\rho}{\sigma}+1)/\ve^4)$ bits of communication to Bob, such that the state $\tilde{\rho}$ that Bob outputs at the end of the protocol satisfies  $\F(\rho,\tilde{\rho})\geq 1-5\ve$ .
\end{theorem}

\begin{proof}
Let the eigen-decomposition of $\rho$ be $\sum_{i=1}^Na_i\ketbra{a_i}$ and that of $\sigma$ be $\sum_{i=1}^Nb_i\ketbra{b_i}$. Define $c\defeq\relent{\rho}{\sigma}$, $\delta \defeq (\ve/3)^4$ and $c' \defeq (c+2)/\delta$. Without loss of generality, assume $a_1,a_2\ldots a_N, \frac{2^{c'}}{\delta}b_1,\frac{2^{c'}}{\delta}b_2\ldots \frac{2^{c'}}{\delta}b_N$ to be rational numbers, and define $K$ be the least common multiple of their denominators. The error due to this assumption can be made arbitrarily close to $0$, for large enough $K$.  

Let $\set{\ket{1},\ket{2}\ldots \ket{N}}$ be an orthonormal basis for $\H$. Introduce registers $A_1,B_1$ associated to $\H$ and registers $A_2,B_2$ associated to some Hilbert space $\H'$ with an orthonormal basis $\set{\ket{1},\ket{2}\ldots \ket{K}}$.
 
Consider the following state on $A_1,A_2,B_1,B_2$. 
\begin{equation}
\label{eq:sharedstate}
\ket{S}_{A_1A_2B_1B_2} \defeq \frac{1}{\sqrt{KN}}\sum_{i=1}^N \ket{i,i}_{A_1B_1} \otimes \left( \sum_{m=1}^{K}\ket{m, m}_{A_2B_2} \right)
\end{equation}

For brevity, define registers $A,B$ such that $A\defeq A_1A_2$ and $B\defeq B_1B_2$.

The protocol is described below.

\begin{mdframed}
\bigskip
\textbf{Input:} Alice is given $\rho=\sum_{i=1}^Na_i\ketbra{a_i}$. Bob is given $\sigma=\sum_{i=1}^Nb_i\ketbra{b_i}$.
\bigskip

%\textbf{Parameters:} $c=\relent{\rho}{\sigma}$, $\delta=(\ve/3)^4$, $c'=(c+2)/\delta$. $a_1,a_2\ldots a_N, \frac{2^{c'}}{\delta}b_1,\frac{2^{c'}}{\delta}b_2\ldots \frac{2^{c'}}{\delta}b_N$ are rational numbers. $K$ is the least common multiple of their denominators.

\textbf{Shared resources:} Alice and Bob hold $\lceil N\log(\frac{1}{\delta}) \rceil$ registers $A_1^iA_2^iB_1^iB_2^i$ ($i \in [\lceil N\log(\frac{1}{\delta}) \rceil]$), such that $A_1^i\equiv A_1, A^i_2\equiv A_2, B_1^i\equiv B_1,B_2^i\equiv B_2$. The shared state in register $A_1^iA_2^iB_1^iB_2^i$ is $\ket{S}_{A_1^iA_2^iB_1^iB_2^i}$. Let $i$ refer to the `index' of corresponding registers.

They also share infinitely many random hash functions $h_1,h_2,\cdots$, where each $h_l:\set{0,\cdots,N-1}\rightarrow\set{0,1}$.

\begin{enumerate}

\item \textbf{For} $i=1$ to $\lceil N\log(\frac{1}{\delta}) \rceil$,

\begin{enumerate}

\item Alice performs the measurement $\set{P_A,I_A- P_A}$ on each register $A_1^iA^i_2$ where,
\begin{equation}
\label{eq:aliceproj}
P_A\defeq\sum_i\ket{a_i}_{A_1}\bra{a_i}_{A_1}\otimes\left(\sum_{m=1}^{Ka_i}\ket{m}_{A_2}\bra{m}_{A_2}\right) .
\end{equation}

On each index $i$, she declares \textit{success} if her outcome corresponds to $P_A$.

\item Bob performs the measurement $\set{P_B,I_B- P_B}$ on each register $B_1^iB^i_2$ where,
\begin{equation}
\label{eq:bobproj}
P_B\defeq \sum_i\ket{b_i}_{B_1}\bra{b_i}_{B_1}\otimes\left(\sum_{m=1}^{\min\{\frac{K}{\delta} 2^{c'} b_i,K\}}\ket{m}_{B_2}\bra{m}_{B_2}\right) .
\end{equation}

On each index $i$, he declares \textit{success} if his outcome corresponds to $P_B$.
\end{enumerate}

\textbf{Endfor}

\item If Alice does not succeed on any index, she aborts.

\item Else, Alice selects the first index $m$ where she succeeds and sends to Bob the binary encoding of $k=\lceil m/N\rceil$ using $\lceil\log\log\frac{1}{\delta}\rceil$ bits.  

\item Alice sends $\{ h_l(m\mod N) |~ l \in  [\lceil r + 2\log\frac{1}{\delta}\rceil]\}$  to Bob.

\item  Define $S_B \defeq \{ t |~ \text{ Bob succeeds on index } $t$\} \cap \{(k-1)N,  \cdots, k N -1\}$. If $S_B$ is empty, he outputs $\ketbra{0}$. Bob selects the first index $n$ in $S_B$ such that $\forall l \in  [\lceil r + 2\log\frac{1}{\delta}\rceil]: h_l(n \mod N) = h_l(m \mod N) $ and outputs the state in $B_1^n$ (if no such index exists, he outputs $\ketbra{0}$).  

\end{enumerate}

\end{mdframed}
\bigskip

We analyze the protocol through a series of claims. 
Following claim computes the probability of success for Alice and Bob.

\begin{claim}
\label{claim:successprob}
For each index $i$, $\prob{\text{Alice succeeds}}=\frac{1}{N}$; $\prob{\text{Bob succeeds}}\leq\frac{2^{c'}}{\delta N}$
\end{claim}

\begin{proof}
Follows from direct calculation.
\end{proof}

From quantum substrate theorem (Fact \ref{fact:substrate}), there exists a state $\rho'$ which satisfies $\F(\rho,\rho')\geq 1-\delta$ and $$\rminent{\rho'}{\sigma}\leq\frac{\relent{\rho}{\sigma}+1}{\delta}+\log\frac{1}{1-\delta}\leq \frac{\relent{\rho}{\sigma}+2}{\delta}=c'.$$ We prove the following claim which is of independent interest as well. 
\begin{claim}\label{claim:1}
Let $\rho'$ have the eigen-decomposition $\rho' = \sum_i g_i\ket{g_i}\bra{g_i}$. For any $p>0$ and every $\ket{g_i}\bra{g_i}$, we have $\sum_{j|~b_j\leq p\cdot g_i}\abs{\braket{b_j}{g_i}}^2 \leq 2^{c'}\cdot p$.
\end{claim}
\begin{proof}
Since  $\rho' \leq 2^{c'}\sigma$, it implies $g_i\ket{g_i}\bra{g_i} \leq 2^{c'}\sigma$. Let $\Pi$ be the projection onto the eigen-space of $\sigma$ with eigenvalues less than or equal to $p\cdot g_i$. We have $\Pi\sigma\Pi \leq p\cdot g_i\cdot\Pi$. After applying $\Pi$ on both sides of the equation $g_i\ket{g_i}\bra{g_i} \leq 2^{c'}\sigma$ and taking operator norm on both sides, we get $g_i\sum_{j:~b_j \leq p\cdot g_i}\abs{\braket{b_j}{g_i}}^2 \leq 2^{c'}\cdot p \cdot g_i $. This implies the lemma. 
\end{proof}

Define 
$$\ket{S_A(\rho)}\defeq \frac{1}{\sqrt{K}}\sum_{i=1}^N\ket{a_i}\ket{\conjugate{a_i}}\otimes \left(\sum_{m=1}^{Ka_i}\ket{m, m}\right) ; $$
$$\ket{S_A(\rho')}\defeq\frac{1}{\sqrt{K}}\sum_{i=1}^N\ket{g_i}\ket{\conjugate{g_i}}\otimes \left(\sum_{m=1}^{\lceil Kg_i \rceil}\ket{m, m} \right) .$$
Here $\ket{\conjugate{a_i}}$ (similarly $\ket{\conjugate{g_i}}$) is the state obtained by taking complex conjugate of  $\ket{a_i}$ ($\ket{g_i}$), with respect to the basis $\set{\ket{1},\ket{2}\ldots \ket{N}}$ in $\H$.  

The following claim asserts that $\ket{S_A(\rho)}$ and $\ket{S_A(\rho')}$ are close if $\rho$ and $\rho'$ are close.
\begin{claim}\label{claim:2}
$\abs{\braket{S_A(\rho)}{S_A(\rho')}} \geq 1- 2 (1-\F(\rho,\rho'))^{1/4} .$
\end{claim}
\begin{proof}
Define $R_{ij}\defeq a_i \abs{\braket{a_i}{g_j}}^2$ and $R'_{ij}\defeq g_i \abs{\braket{a_i}{g_j}}^2$. Note that both $R \defeq \{R_{ij}\}$ and $R' \defeq \{R'_{ij}\}$ form probability distributions over $[N^2]$. Also note that $\F(R,R') = \Tr(\sqrt{\rho}\sqrt{\rho'})$. Consider
\begin{align*}
 \abs{\braket{S_A(\rho)}{S_A(\rho')}} & =\sum_{i,j}\min(R_{ij},R'_{i,j}) = 1 - \frac{1}{2}\onenorm{R-R'} \nonumber\\
 & \geq 1 - \sqrt{1-\F(R,R')^2} = 1 - \sqrt{1-(\Tr \sqrt{\rho}\sqrt{\rho'})^2} \nonumber \\
 & \geq 1 - \sqrt{2(1- \Tr \sqrt{\rho}\sqrt{\rho'})} \geq  1- \sqrt{2 \sqrt{1-\F(\rho,\rho')^2}}  \nonumber \\
 & \geq 1- 2 (1-\F(\rho,\rho'))^{1/4} . \myqedhere 
  \end{align*}
where the first equality is from the definitions of $\ket{S_A(\rho)}$ and $\ket{S_A(\rho')}$; the second equality is from the definition of $\ell_1$ distance; the first inequality is from ~\ref{fact:tracefidelityequi}; the second inequality is from the fact that $\Tr\sqrt{\rho}\sqrt{\rho'}\leq 1$; the third inequality is from Facts~\ref{fact:fidelityvstrace}.
\end{proof}
We use these claims to prove the following.
\begin{claim}
\label{claim:relativeprob}
For each index $i$, $\prob{\text{Bob succeeds}|~\text{Alice succeeds}}\geq 1 - \delta -  2 \delta^{1/4} \geq 1 - \ve $.
\end{claim}

\begin{proof}
Consider,
\begin{equation*}
(I_A\otimes P_B)\ket{S_A(\rho')} = \frac{1}{\sqrt{K}}\sum_{i,j=1}^N\ket{\conjugate{g_j}} \ket{b_i}\braket{b_i}{g_j} \left(\sum_{m=1}^{\min\{\lceil Kg_j\rceil, \frac{K}{\delta} 2^{c'}b_i\}}\ket{m, m} \right).
\end{equation*}
Therefore,
\begin{align}
& \norm{(I_A\otimes P_B)\ket{S_A(\rho')} }^2 \geq \sum_{i,j=1}^N \abs{\braket{b_i}{g_j}}^2 \min\{g_j, \frac{1}{\delta}2^{c'}b_i \} \nonumber \\
& \geq \sum_{j=1}^N g_j \left(\sum_{i|~b_i\geq \delta 2^{-c'}g_j} \abs{\braket{b_i}{g_j}}^2 \right) \geq \sum_{j=1}^N g_j(1- \delta) = 1- \delta  . \quad \mbox{(using Claim~\ref{claim:1})} \label{eq:probonrho'}
\end{align}
Using the above,

\begin{align*}
& \prob{\text{Bob succeeds}|~\text{Alice succeeds}}  =  \Tr (I_A\otimes P_B)\ketbra{S_A(\rho)} \\
& \geq  \Tr (I_A\otimes P_B)\ketbra{S_A(\rho')}  - \frac{1}{2} \onenorm{S_A(\rho) - S_A(\rho')} \\
& = \Tr (I_A\otimes P_B)\ketbra{S_A(\rho')}  -  \sqrt{1- \abs{\braket{S_A(\rho)}{S_A(\rho')}}^2} \quad \mbox{(Fact~\ref{fact:tracefidelityequi})} \\
& \geq  1 -  \delta  - 2  \sqrt{ (1-\F(\rho,\rho'))^{1/2}}  . \quad \mbox{(Claim~\ref{claim:2} and Eq.~\eqref{eq:probonrho'})}
\end{align*}
\end{proof}

Finally, we show that if Alice and Bob succeed together on an index, the state in register $B$ with Bob is close to $\rho$.
\begin{claim}
\label{claim:boboutput}
Given that both Alice and Bob succeed, fidelity between $\rho$ and the state of the register $B$ is at least $\sqrt{1 - \delta -  2 \delta^{1/4}} \geq 1 - \ve $ .
\end{claim}

\begin{proof}
From gentle measurement lemma (Fact~\ref{fact:gentlemeasurement}), 
$$\F(S_A(\rho), \frac{(I_A\otimes P_B)\ketbra{S_A(\rho)} (I_A\otimes P_B) }{\Tr (I_A\otimes P_B)\ketbra{S_A(\rho)}}) \geq \sqrt{\Tr (I_A\otimes P_B)\ketbra{S_A(\rho)}}. $$ 
Since the marginal of $\ket{S_A(\rho)}$ on  register $B$ is $\rho$ and partial trace does not decrease fidelity (Fact~\ref{fact:monotonequantumoperation}), using item 2. above, the desired result follows. 

\end{proof}

Let $j$ be the first index where Alice and Bob both succeed. As described in the protocol, $m$ is the first index where Alice succeeds and $n$ is the index such that Bob outputs the state in $B_1^n$. 
We have the following claim, 

\begin{claim}
\label{claim:highprobindex}
With probability at least $1-4\ve$, $m=n=j$.
\end{claim}

Before proving Claim~\ref{claim:highprobindex}, let us define the following "bad" events.
\begin{definition}\label{def:badevent}
\begin{itemize}

\item $T_1$ is the event that Alice does not succeed on any of the indices.

\item $T_2$ is the event that $m \notin S_B$ conditioned on $\neg T_1$ .
\item $T_3$ represents the event that $n \neq m$ conditioned on $\neg T_1$.
\end{itemize}
\end{definition}

Notice that if none of above events occur, then both Alice and Bob output the same index $n=m$, and since $m$ is the first index at which Alice succeeds, $n=m=j$. 

We have the following claim.

\begin{claim}\label{claim:protocol}
It holds that: \quad $1. ~ \prob{T_1}\leq \ve; \quad 2. ~ \prob{T_2}\leq\ve; \quad 3. ~ \prob{T_3}\leq3\ve.$
\end{claim}
\begin{proof}

\begin{enumerate}

\item $\prob{T_1}\leq\br{1-\frac{1}{N}}^{\lceil N \cdot \log\frac{1}{\ve}\rceil}\leq \exp^{-\lceil\log\frac{1}{\ve}\rceil}\leq  \delta$.

\item Follows from Claim \ref{claim:relativeprob}.

\item For this argument we condition on $\neg T_1$ for all events below. From Claim \ref{claim:successprob} and the fact that Bob independently measures each index, we have  $\expec{}{\abs{S_B}}= N\cdot \prob{\text{Bob succeeds}} \leq \frac{2^{c'}}{\delta}$. Using Markov's inequality, 
\begin{equation} \label{eq:sb} 
\prob{\abs{S_B} \geq \frac{2^{c'}}{\delta\ve}}\leq \frac{\delta\ve}{2^{c'}} \cdot \expec{}{\abs{S_B}} \leq \ve . 
\end{equation}

    Thus
    \begin{align*}
    \prob{T_3} & \leq\prob{\abs{S_B}\geq \frac{2^{c'}}{\delta\ve} \text{ or } m \notin S_B} +  \prob{T_3~|~m \in S_B \text{ and } \abs{S_B}\leq \frac{2^{c'}}{\delta\ve}} \\ & \leq \prob{\abs{S_B}\geq \frac{2^{c'}}{\delta\ve}} + \prob{T_2} + \prob{T_3~|~m \in S_B \text{ and } \abs{S_B}\leq \frac{2^{c'}}{\delta\ve}} \\
     & \leq 2 \ve+\prob{T_3~|~m \in S_B \text{ and } \abs{S_B}\leq \frac{2^{c'}}{\delta\ve}}  \quad \mbox{(Eq.~\eqref{eq:sb} and item 2. of this claim)}\\
     & \leq 2\ve + 2^{- \lceil c' + \log\frac{1}{\delta}+ 2\log\frac{1}{\ve}\rceil} \cdot \frac{2^{c'}}{\delta\ve} \quad  \leq \quad 3 \ve.
    \end{align*}

\end{enumerate}
\end{proof}

We bound the probability that $m\neq n$. If $m=n$, then $m$ being the first index on which Alice succeeds, we have $m=n=j$. 

\bigskip

\noindent{\em Proof of Claim\ref{claim:highprobindex}}.\quad We conclude the claim since,
$$\prob{n \neq m} \leq\prob{T_1}+\prob{\neg T_1}\cdot\prob{T_3}\leq 4\ve.$$ \qed

From claims \ref{claim:successprob},\ref{claim:relativeprob} and \ref{claim:highprobindex}, the probability that Bob learns the index $j$ is atleast $1-4\ve$. Conditioned on this event, Claim \ref{claim:boboutput}, implies that the state $\rho' \in \mathcal{D}(B^j)$ that Bob outputs satisfies $\F(\rho',\rho)\geq 1-\ve$. Conditioned on the event that Bob learns the wrong index or the protocol is aborted, let the state output by Bob be $\mu$. Then Bob outputs the state $\tilde{\rho}=\alpha\rho'+(1-\alpha)\mu$, where $\alpha\geq 1-4\ve$. Using concavity of fidelity (Fact \ref{fact:fidconcave}), we have $\F(\tilde{\rho},\rho)\geq \alpha\F(\rho',\rho)+(1-\alpha)\F(\mu,\rho) \geq (1-4\ve)(1-\ve)\geq 1- 5\ve$. 

The communication cost of above protocol is $$\lceil\log\log\frac{1}{\ve}\rceil+\lceil c' + \log\frac{1}{\delta} + 2\log\frac{1}{\ve}\rceil\leq \lceil 3^4\frac{c+2}{\ve^4}+7\log\frac{1}{\ve} \rceil.$$ This completes the proof of theorem.     
\end{proof}

It may be noted that variants of the part of protocol that uses hash functions, have appeared in many other works such as~\cite{Braverman2011,Kerenidis2012}.

\begin{remark}
\label{qbevremark}

Note that if Alice and Bob get a real number $r> \relent{\rho}{\sigma}$, instead of $\relent{\rho}{\sigma}$ (all other inputs remaining the same), the protocol above works in the same fashion, with the communication upper bounded by $O((r+1)/\ve^4)$.
\end{remark}

\subsection{Compression with side information}
Here we present a variant of our protocol with side information. We start with the following.
\begin{lemma}
\label{lem:withsideinf}
Let $A,B$ be two registers. Alice is given the eigen-decomposition of a bipartite state $\rho_{AB} \in \mathcal{D}(AB)$. Bob is given the eigen-decompositions of a bipartite state $\sigma_{AB}\in \mathcal{D}(AB)$ and the state $\rho_A\defeq \Tr_{B}(\rho_{AB})$, such that $\text{supp}(\rho_{AB})\subset \text{supp}(\sigma_{AB})$. Define $\sigma_A\defeq\Tr_{B}(\sigma_{AB})$. Let  $\relent{\rho_{AB}}{\sigma_{AB}}-\relent{\rho_A}{\sigma_A}$ and $\ve >0$ be known to Alice and Bob. There exists a protocol, in which Alice and Bob use shared entanglement and Alice sends $\O((\relent{\rho_{AB}}{\sigma_{AB}}-\relent{\rho_A}{\sigma_A}+1)/\ve^4)$ bits of communication to Bob such that the state $\tilde{\rho}_{AB}$ that Bob outputs at the end of the protocol satisfies  $\F(\rho_{AB},\tilde{\rho}_{AB})\geq 1-5\ve$ .
\end{lemma}

\begin{proof}
Following equality follows from definitions. 
\begin{equation*}
\relent{\rho_{AB}}{\sigma_{AB}}-\relent{\rho_A}{\sigma_A} = \relent{\rho_{AB}}{e^{\text{log}(\sigma_{AB})-\text{log}(\sigma_{A}\otimes I_B)+\text{log}(\rho_{A}\otimes I_B)}}  \enspace .
\end{equation*}
Define, 
$$Z= \Tr(e^{\text{log}(\sigma_{AB})-\text{log}(\sigma_{A}\otimes I_B)+\text{log}(\rho_{A}\otimes I_B)}) \quad ; \quad \tau_{AB}=e^{\text{log}(\sigma_{AB})-\text{log}(\sigma_{A}\otimes I_B)+\text{log}(\rho_{A}\otimes I_B)}/Z \enspace .$$ 
It holds that $Z\leq 1$ (from Fact \ref{fact:ruskai}) and hence $\relent{\rho_{AB}}{\tau_{AB}} \leq \relent{\rho_{AB}}{\sigma_{AB}}- \relent{\rho_A}{\sigma_A}$. 
Bob computes the eigen-decomposition of $\tau_{AB}$ using his input. They run the protocol given by Theorem~\ref{thm:protocol} with the following setting: Alice knows a state $\rho_{AB}$, Bob knows a state $\tau_{AB}$ and both know a number
($=\relent{\rho_{AB}}{\sigma_{AB}}- \relent{\rho_A}{\sigma_A}$) greater than $\relent{\rho_{AB}}{\tau_{AB}}$.
 They also know the error parameter $\ve>0$. By the virtue of Remark \ref{qbevremark}, at the end of the protocol, Bob obtains a state $\tilde{\rho}_{AB}$, such that $\F(\rho_{AB},\tilde{\rho}_{AB})\geq 1-5\ve$. Communication from Alice is upper bounded by $\O(\br{\relent{\rho_{AB}}{\sigma_{AB}}-\relent{\rho_A}{\sigma_A} +1}/\epsilon^4)$. 
\end{proof}

We now present the protocol $\Prot'$ as mentioned in the Introduction. 
\begin{theorem}
\label{thm:sideinfchannel}
Let $A,B$ be two registers associated to Hilbert spaces $\H_A,\H_B$ respectively. Alice and Bob know a Stinespring representation (Fact \ref{stinespring}) of a quantum channel $\E: \mathcal{L}(A)\rightarrow \mathcal{L}(B)$. Alice is given the eigen-decomposition of a state $\rho \in \mathcal{D}(A)$. Bob is given the eigen-decompositions of a state $\sigma \in \mathcal{D}(A)$ (such that $\text{supp}(\rho)\subset \text{supp}(\sigma)$) and the state $\rho'= \E(\rho)$. Let  $\relent{\rho}{\sigma}-\relent{\E(\rho)}{\E(\sigma)}$ and $\ve >0$ be known to Alice and Bob. There exists a protocol, in which Alice and Bob use shared entanglement and Alice sends $\O((\relent{\rho}{\sigma}-\relent{\E(\rho)}{\E(\sigma)}+1)/\ve^4)$ bits of communication to Bob, such that the state $\tilde{\rho}$ that Bob outputs at the end of the protocol satisfies  $\F(\rho,\tilde{\rho})\geq 1-5\ve$ .
\end{theorem}

\begin{proof}
Let a Stinespring representation of $\E$ be 
$\E(\omega) = 
\Tr_{A,C}\br{V (\omega \ketbra{0}_{BC})V^{\dagger}}$, where  $V:\H_A \otimes \H_B \otimes \H_C \rightarrow \H_A \otimes \H_B\otimes \H_C$ is a unitary operation (Fact~\ref{stinespring}). Alice and Bob compute the states $V\br{\rho\otimes\ketbra{0}_{BC}} V^{\dagger}$ and $V\br{\sigma\otimes\ketbra{0}_{BC}} V^{\dagger}$ , respectively. From Lemma~\ref{lem:withsideinf} and the equality $\relent{V\br{\rho\otimes\ketbra{0}_{BC}}V^{\dagger}}{V\br{\sigma\otimes\ketbra{0}_{BC}} V^{\dagger}}=\relent{\rho}{\sigma}$, there exists a protocol, in which Alice and Bob use shared entanglement and Alice sends $\O(\relent{\rho}{\sigma}-\relent{\E(\rho)}{\E(\sigma)}+1)/\ve^4 $ bits of communication to Bob, such that the state $\tilde{\rho}_{ABC}$ that Bob gets at the end of the protocol satisfies  $\F(V\br{\rho\otimes\ketbra{0}_{BC}} V^{\dagger},\tilde{\rho}_{ABC})\geq1-5\ve$.  Bob outputs $\tilde{\rho} = \Tr_{BC} V^{\dagger}\br{\tilde{\rho}_{ABC}} V $. From monotonicity of fidelity under quantum operation (Fact~\ref{fact:monotonequantumoperation}), $\F(\rho,  \tilde{\rho} )\geq1-5\ve$. 
\end{proof}

\section{A direct sum theorem for quantum one-way communication complexity}
\label{sec:directsum}

As a consequence of Theorem~\ref{thm:protocol} we obtain the following direct sum result for all relations in the model of entanglement-assisted one-way communication complexity. 

\begin{theorem}\label{thm:directsumoneway}
Let $\X,\Y,\Z$ be finite sets, $f\subseteq\X\times\Y\times\Z$ be a relation, $0<\ve,\delta$ be error parameters and $k>1$ be an integer. We have
\[\dqcc{}{\ve}{f^k}\geq \Omega \br{k \br{\delta^9\cdot\dqcc{}{\ve+\delta}{f}-1}}.\]
\end{theorem}

\begin{proof}
Let  $\mu$ be any distribution over $\X\times\Y$. We show the following, which combined with Fact~\ref{fact:yaos principle} implies the desired:
\[\dqcc{, \mu^k}{\ve}{f^k}\geq\Omega \br{k \br{\delta^9\cdot\dqcc{, \mu}{\ve+\delta}{f}-1}}.\]

Let $\Q$ be a quantum one-way protocol with communication $c\cdot k$ computing $f^k$  with  overall probability of success at least $1-\ve$ under distribution $\mu^k$. Let the inputs to Alice and Bob be given in registers $X_1,X_2\ldots X_k$ and $Y_1,Y_2\ldots Y_k$. For brevity, we define $X\defeq X_1,X_2\ldots X_k$ and $Y\defeq Y_1,Y_2\ldots Y_k$. Thus, the state $\sum_{xy}\mu^k(x,y)\ketbra{xy}_{XY}$ represents the joint input, where $x$ is drawn from $X$ and $y$ is drawn from $Y$.

Let $\sigma_{E_A,E_B}$ be the shared entanglement between Alice and Bob where register $E_A$ is with Alice and $E_B$ with Bob. Alice applies unitary $U:\H_X\otimes \H_{E_A}\rightarrow \H_X\otimes \H_A\otimes \H_M$, where $E_A\equiv AM$, sends the message register $M$ to Bob, and then Bob applies the unitary $V: \H_Y\otimes\H_M\otimes \H_{E_B}\rightarrow \H_Y\otimes\H_{B'}\otimes\H_Z$, where $ME_B\equiv B'Z$. Since unitary operations by Alice and Bob are conditioned on their respective inputs, the unitaries $U,V$ are  of the form $U=\sum_x\ketbra{x}_X\otimes U_x$ and $V=\sum_y\ketbra{y}_Y\otimes V_y$, where $U_x:\H_{E_A}\rightarrow \H_A\otimes \H_M$ and $V_y: \H_M\otimes \H_{E_B}\rightarrow \H_{B'}\otimes\H_Z$.
Let the following be the global state before Alice applies her unitary:
$$\theta_{XYE_AE_B} = \sum_{xy}\mu^k(x,y)\ketbra{xy}_{XY}\otimes\sigma_{E_AE_B}.$$ 

Let $D=D_1\cdots D_k$ be a random variable uniformly distributed over $\set{0,1}^k$ and independent of the input $XY$. Define random variables $U_1,U_2\ldots U_k$ such that $U_i=X_i$ if $D_i=0$ and $U_i=Y_i$ if $D_i=1$. Let $U=U_1,U_2\ldots U_k$. Consider the state $\theta_{XYE_AE_BDU}$, with registers $D,U$ as defined above.

 Let $\rho_{XYAME_BDU}\defeq U\theta_{XYE_AE_BDU} U^{\dagger}$ be the state after Alice applies her unitary and sends $M$ to Bob. Since $$\condmutinf{XE_AE_B}{Y}{DU}_{\theta}=0,$$ it holds that  
\[\condmutinf{XAE_BM}{Y}{DU}_{\rho}=0.\]
From the definition of $DU$, we thus have (below $-i$ represents the set $\{1,2\ldots i-1,i+1\ldots k\}$),
\begin{equation}\label{eqn:dumarkov}
\condmutinf{X_{-i}AE_BM}{Y}{X_iD_{-i}U_{-i}}_{\rho}=\condmutinf{XAE_BM}{Y_{-i}}{Y_iD_{-i}U_{-i}}_{\rho}=0.
\end{equation}
Since $\log\abs{M}\leq ck$ and register $E_B$ is independent of registers $XYDU$ in the state $\rho_{E_BXYDU}$, we have
\begin{eqnarray*}
\mutinf{XYDU}{ME_B}_{\rho}&=&\mutinf{XYDU}{E_B}_{\rho}+\condmutinf{XYDU}{M}{E_B}_{\rho}\\ &=& \condmutinf{XYDU}{M}{E_B}_{\rho} \leq 2 \log\abs{M} \leq 2ck,
\end{eqnarray*}
where the second last inequality is from Fact~\ref{fact:conditionalmutinf}.
Consider
\begin{eqnarray*}
2ck&\geq& \mutinf{XYDU}{ME_B}_{\rho}\geq \condmutinf{XY}{ME_B}{DU}_{\rho}\\ &\geq&\sum_{i=1}^k\condmutinf{X_iY_i}{ME_B}{DU}_{\rho} \quad (\text{Fact} ~\ref{fact:chainrulemutualinf})\\
&=& \sum_{i=1}^k \condmutinf{X_iY_i}{ME_B}{D_iU_iD_{-i}U_{-i}}_{\rho} 	\\ &=&\frac{1}{2}\br{\sum_{i=1}^k\br{\condmutinf{X_i}{ME_B}{Y_iD_{-i}U_{-i}}_{\rho}+\condmutinf{Y_i}{ME_B}{X_iD_{-i}U_{-i}}_{\rho}}}\\
&\geq&\frac{1}{2}\sum_{i=1}^k\condmutinf{X_i}{ME_B}{Y_iD_{-i}U_{-i}}_{\rho}.
\end{eqnarray*}
where the last equality is from the definition of $DU$ and the last inequality is from Fact~\ref{fact:strongsubadditivity}.
%But observe that
%$$\condmutinf{Y_i}{ME_B}{X_iD_{-i}U_{-i}}_{\rho}\leq\condmutinf{Y_i}{MAE_B}{X_iD_{-i}U_{-i}}_{\rho}=0,$$
%where inequality is from Fact~\ref{fact:strongsubadditivity} and the last equality is from~(\ref{eqn:dumarkov}).
Hence there exists $j\in[k]$ such that
\begin{equation}\label{eqn:directsummessagelength}
\condmutinf{X_j}{ME_B}{Y_jD_{-j}U_{-j}}_{\rho}\leq 4c.
\end{equation}
Furthermore, we have 
\begin{equation}\label{eqn:directsumpubliccoins}
\mutinf{X_jY_j}{D_{-j}U_{-j}}_{\rho}= \mutinf{X_jY_j}{D_{-j}U_{-j}}_{\theta} = 0.
\end{equation}
since the unitary by Alice does not change the state on registers $DUXY$.

For brevity, set $B\defeq ME_B$. Define the following states, which are obtained by conditioning on various classical registers: $$\rho_B^{x_jy_jd_{-j}u_{-j}}\defeq \frac{\bra{x_jy_jd_{-j}u_{-j}}\rho_{BXYDU}\ket{x_jy_jd_{-j}u_{-j}}}{\bra{x_jy_jd_{-j}u_{-j}}\rho_{XYDU}\ket{x_jy_jd_{-j}u_{-j}}},$$

$$\rho_B^{x_jd_{-j}u_{-j}}\defeq \frac{\bra{x_jd_{-j}u_{-j}}\rho_{BXDU}\ket{x_jd_{-j}u_{-j}}}{\bra{x_jd_{-j}u_{-j}}\rho_{XDU}\ket{x_jd_{-j}u_{-j}}}$$

$$\rho_B^{y_jd_{-j}u_{-j}}\defeq \frac{\bra{y_jd_{-j}u_{-j}}\rho_{BYDU}\ket{y_jd_{-j}u_{-j}}}{\bra{y_jd_{-j}u_{-j}}\rho_{YDU}\ket{y_jd_{-j}u_{-j}}}$$

From~(\ref{eqn:dumarkov}), we have
\[\condmutinf{Y}{B}{X_jD_{-j}U_{-j}}_{\rho}=0.\]
which is equivalent to, using Fact \ref{fact:cqcondmut} and the fact that registers $X,Y,U,D$ are classical in $\rho_B$: 
$$\expec{x_jy_jd_{-j}u_{-j}\leftarrow X_jY_jD_{-j}U_{-j}}{\relent{\rho_B^{x_jy_jd_{-j}u_{-j}}}{\rho_B^{x_jd_{-j}u_{-j}}}}=0.$$ This implies $\rho_B^{x_jy_jd_{-j}u_{-j}}=\rho_B^{x_jd_{-j}u_{-j}}$ for all $x_j,y_j,d_{-j}u_{-j}$. 

From~(\ref{eqn:directsummessagelength}), and Fact \ref{fact:cqcondmut},
$$\expec{x_jy_jd_{-j}u_{-j}\leftarrow X_jY_jD_{-j}U_{-j}}{\relent{\rho_B^{x_jy_jd_{-j}u_{-j}}}{\rho_B^{y_jd_{-j}u_{-j}}}}\leq 4c.$$ 

Let 
$$G\defeq\set{(x_j,y_j,d_{-j},u_{-j}):\relent{\rho_B^{x_jy_jd_{-j}u_{-j}}}{\rho_B^{y_jd_{-j}u_{-j}}}\leq \frac{4c}{\delta}}.$$ By Markov's inequality, 
$$\prob{X_jY_jD_{-j}U_{-j}\in G}\geq1-\delta. $$

Now, we exhibit an entanglement-assisted one-way protocol $\Q'$ for $f$ with communication less than $c$ and distributional error $\ve$ under distribution $\mu$. 
\begin{mdframed}
\bigskip
\begin{enumerate}
\item Alice and Bob share public coins according to distribution $\rho_{D_{-j}U_{-j}}$, and the shared entanglement needed to run the protocol $\P$ from Theorem ~\ref{thm:protocol}.
\item Alice and Bob are given the input  $(x,y)\sim\mu$. They embed the input to the $j$-th coordinate $X_jY_j$. The input is independent of shared randomness, from equation (\ref{eqn:directsumpubliccoins}). 
\item Given input $(x_j,y_j)\equiv (x,y)$ and shared public coins $d_{-j}u_{-j}$, Alice knows the eigen-decomposition of the state $\rho_B^{x_jy_jd_{-j}u_{-j}}$, since $\rho_B^{x_jy_jd_{-j}u_{-j}}=\rho_B^{x_jd_{-j}u_{-j}}$. Bob knows the eigen-decomposition of state $\rho_{B}^{y_jd_{-j}u_{-j}}$.
\item They  run the protocol in Theorem ~\ref{thm:protocol} with inputs $\rho_B^{x_jy_jd_{-j}u_{-j}}$, $\frac{4c}{\delta}$ (given to Alice) and $\rho_B^{y_jd_{-j}u_{-j}}$, $\frac{4c}{\delta}$ (given to Bob). After communicating $\mathcal{O}(4c/\delta^9)$  bits to Bob, Bob receives a state $\sigma_B^{x_jy_jd_{-j}u_{-j}}$ satisfying $\|\sigma_B^{x_jy_jd_{-j}u_{-j}}-\rho_B^{x_jy_jd_{-j}u_{-j}}\|_1\leq \delta$ if $(x_j,y_j,d_{-j},u_{-j})\in G$. 
\item Bob samples the distribution from $\rho_{Y_{-j}}$, since he has the registers $D_{-j}U_{-j}Y_j$. This is possible from equation \ref{eqn:dumarkov}, which states that register $Y_{-j}$ is independent of registers $A,B,X$ conditioned on registers $D_{-j}U_{-j}Y_j$. 
\item Bob applies the unitary $V$, as in the protocol $\Q$, on registers $BY\equiv E_BMY$ and then measures the register $Z$. He outputs the outcome. 
\end{enumerate}
\end{mdframed}
\bigskip
From the protocol, it is clear that overall distributional error in $\Q'$ is at most $2\delta+\ve$. The error $2\delta$ occurs since the state $\sigma_B^{x_jy_jd_{-j}u_{-j}}$ satisfies $\|\sigma_B^{x_jy_jd_{-j}u_{-j}}-\rho_B^{x_jy_jd_{-j}u_{-j}}\|_1\leq \delta$ and the probability that $(x_j,y_j,d_{-j},u_{-j})\notin G$ is at most $\delta$. The error $\ve$ is due to the original protocol $\Q$.
Hence 
\[\dqcc{, \mu}{\ve+2\delta}{f}\leq \O((c+1)/\delta^9),\]
which implies (changing $\delta\rightarrow \frac{\delta}{2}$)
\[\dqcc{, \mu^k}{\ve}{f^k}\geq\Omega \br{k \br{\delta^9\cdot\dqcc{, \mu}{\ve+\delta}{f}-1}}.\]
\end{proof}

\section{Quantum correlated sampling}

In this section, we give a quantum analogue to classical correlated sampling. In our framework, Alice and Bob (given quantum states $\rho$ and $\sigma$ respectively as inputs) create a joint quantum state with marginals $\rho$ and $\sigma$ on respective sides. The joint state has the property that same projective measurement performed by Alice and Bob gives very correlated outcomes, if $\rho$ and $\sigma$ are close to each other in $\ell_1$ distance.  
 Following theorem makes this sampling task precise.  
\label{sec:qcorr}
\begin{theorem}\label{thm:correlatedsampling}
Let $\rho,\sigma$ be quantum states on a Hilbert space $\H$ of dimension $N$. Alice is given the eigen-decomposition of $\rho$ and Bob is given the eigen-decomposition of $\sigma$.  There exists a zero-communication  protocol satisfying the following.
\begin{enumerate}
\item Alice outputs registers  $A_1,A_2$ and and Bob outputs registers $B_1,B_2$ respectively, such that state in $A_1$ is $\rho$, the state in $B_1$ is $\sigma$ and $A_1\equiv B_1, A_2\equiv B_2$.
\item Let $M = \{M_1,M_2\ldots M_w\}$ be a projective measurement, in the support of $A_1A_2$. Let $M$ be performed by Alice on the joint system $A_1A_2$ with outcome $I \in [w]$ and by Bob on the joint system $B_1B_2$ with outcome $J \in [w]$. Then $\Pr[I = J] \geq \left(1-\sqrt{\onenorm{\rho-\sigma} - \frac{1}{4}\onenorm{\rho-\sigma}^2}\right)^{3}$. 
\end{enumerate}
\end{theorem}

\begin{proof}
Let eigen-decomposition of $\rho$ be $\sum_{i=1}^Na_i\ketbra{a_i}$ and of $\sigma$ be $\sum_{i=1}^Nb_i\ketbra{b_i}$. Let $\set{\ket{1},\ket{2}\ldots \ket{N}}$ be an orthonormal basis for $\H$. We assume that $a_1, \ldots, a_N, b_1,  \ldots b_N$ are rational numbers and let $K$ be the least common multiple of their denominators. The error due to this assumption goes to $0$ as $K \rightarrow \infty$. 

Introduce registers $A_1,B_1$ associated to $\H$ and registers $A_2,B_2$ associated to some Hilbert space $\H'$ with an orthonormal basis $\set{\ket{1},\ket{2}\ldots \ket{K}}$.
 
Consider the following state shared in $A_1,A_2,B_1,B_2$.
\begin{equation*}
\ket{S}_{A_1B_1A_2B_2} \defeq \frac{1}{\sqrt{KN}}\sum_{i=1}^N \ket{i,i}_{A_1B_1} \otimes \left(\sum_{m=1}^{K}\ket{m, m}_{A_2B_2} \right)
\end{equation*}

 For brevity, define the registers $A\defeq A_1A_2$ and $B \defeq B_1B_2$. The protocol is described below.
\begin{mdframed}
\bigskip
\textbf{Input:} Alice is given $\rho=\sum_{i=1}^Na_i\ketbra{a_i}$. Bob is given $\sigma=\sum_{i=1}^Nb_i\ketbra{b_i}$.
\bigskip

\textbf{Shared resources:} Alice and Bob hold infinitely many registers $A_1^iA_2^iB_1^iB_2^i$ ($i>0$), such that $A_1^i\equiv A_1, A^i_2\equiv A_2, B_1^i\equiv B_1,B_2^i\equiv B_2$. The shared state in register $A_1^iA_2^iB_1^iB_2^i$ is $\ket{S}_{A_1^iA_2^iB_1^iB_2^i}$. Let $A \equiv A_1A_2$ and $B\equiv B_1B_2$ be used as output registers. Let $i$ refer to the `index' of corresponding registers.

\begin{enumerate}

\item For each $i>0$, Alice performs the measurement $\set{P_A,I- P_A}$ on the registers $A^i_1A^i_2$, where
\begin{equation*}
P_A\defeq\sum_i\ket{a_i}_{A_1}\bra{a_i}_{A_1}\otimes \left(\sum_{m=1}^{Ka_i}\ket{m}_{A_2}\bra{m}_{A_2}\right)
\end{equation*}

She declares \textit{success} if she obtains outcome corresponding to $P_A$. She stops once she succeeds in some register $A^j$, and swaps $A^j$ with $A$. 

\item For each $i>0$, Bob performs the measurement $\set{P_B,I- P_B}$ on the registers $B^i_1B^i_2$, where
\begin{equation*}
P_B\defeq\sum_i\ket{b_i}_{B_1}\bra{b_i}_{B_1}\otimes \left(\sum_{m=1}^{Kb_i}\ket{m}_{B_2}\bra{m}_{B_2}\right)
\end{equation*}

He declares \textit{success} if he obtains outcome corresponding to $P_B$. He stops once he succeeds in some register $B^j$, and swaps $B^j$ with $B$.

\end{enumerate}

\end{mdframed}
\bigskip

At the end of above protocol, let the joint state in the register $AB$ be $\tau$. The following claim shows the first part of the theorem. 
\begin{claim}\label{redstate}
$\Tr_{A_2B_1B_2}(\tau)=\rho$ and $\Tr_{A_1A_2B_2}(\tau)=\sigma$.
\end{claim}
\begin{proof}
It is easily seen that the marginal of the state $(P_A \otimes I_B) \ket{S}$ in register $A$  is $\rho$. Similarly the marginal of the state $(I_A \otimes P_B) \ket{S}$ in register $B$ in  is $\sigma$.
\end{proof}

Following series of claims establish second part of the theorem.
\begin{claim} \label{claim:tau}
\begin{equation*}
\tau \geq \frac{(P_A\otimes P_B)\ket{S}\bra{S}(P_A\otimes P_B) }{1 - \bra{S} (I_A-P_A)\otimes(I_B-P_B)\ket{S}} \enspace .
\end{equation*}
\end{claim}
\begin{proof}
Consider the event that Alice and Bob succeed at the same index. The resulting state in $AA_1BB_1$ is 
$$\frac{(P_A\otimes P_B ) \ket{S}\bra{S} (P_A\otimes P_B)}{\bra{S}(P_A \otimes P_B)\ket{S}},$$ 
and this event occurs with probability 
$$\sum_{i=0}^{\infty}\bra{S}(I_A-P_A)\otimes(I_B-P_B)\ket{S}^{i} \cdot \bra{S}(P_A \otimes P_B)\ket{S}= \frac{\bra{S}(P_A \otimes P_B)\ket{S}}{1 - \bra{S}(I_A-P_A)\otimes(I_B-P_B)\ket{S}}.$$ 
Since the cases of Bob succeeding before Alice and Alice succeeding before Bob add positive operators to $\tau$, we get the desired.  
\end{proof}

\begin{claim}\label{overlaptau}
Let $\ket{\theta} \defeq \frac{(P_A\otimes P_A)\ket{S}}{\norm{(P_A\otimes P_A)\ket{S}}}$. Then 
$$\bra{\theta}\tau\ket{\theta} \geq \frac{\left(1-\sqrt{\onenorm{\rho-\sigma} - \frac{1}{4}\onenorm{\rho-\sigma}^2}\right)^2}{1 + \sqrt{\onenorm{\rho-\sigma} - \frac{1}{4}\onenorm{\rho-\sigma}^2}} \geq \left(1-\sqrt{\onenorm{\rho-\sigma} - \frac{1}{4}\onenorm{\rho-\sigma}^2}\right)^3 .$$ 
\end{claim}
\begin{proof}
Consider,
\begin{align*}
\bra{\theta}\tau\ket{\theta} &\geq   \frac{ \abs{\bra{\theta}P_A\otimes P_B\ket{S}}^2}{1 - \bra{S} (I_A-P_A)\otimes(I_B-P_B)\ket{S}} &  \mbox{(Claim~\ref{claim:tau})}\\ 
& = \frac{ \abs{\bra{\theta}P_A\otimes P_B\ket{S}}^2}{2/N-\bra{S}P_A\otimes P_B\ket{S}} & \mbox{(using $\bra{S}P_A\otimes I_B\ket{S}=\bra{S}I_A\otimes P_B\ket{S}=1/N$)}.
\end{align*}

By direct calculation, we get 
$$(P_A\otimes P_B)\ket{S}=\frac{1}{\sqrt{KN}}\sum_{i,j}\ket{\conjugate{a_i}}\braket{b_j}{a_i}\ket{b_j}\sum_{m=1}^{K\min(a_i,b_j)}\ket{m,m} ;$$ 
$$\ket{\theta}=\frac{1}{\sqrt{K}}\sum_i\ket{\conjugate{a_i}}\ket{a_i}\sum_{m=1}^{Ka_i}\ket{m,m} .$$ 
Hence,
\begin{equation} \label{eq:thetatau}\bra{\theta}\tau\ket{\theta} \geq \frac{\left(\sum_{i,j}\min(a_i,b_j) \abs{\braket{a_i}{b_j}}^2\right)^2}{2-\sum_{i,j}\min(a_i,b_j) \abs{\braket{a_i}{b_j}}^2} .
\end{equation}

Define $R_{ij}\defeq a_i \abs{\braket{a_i}{b_j}}^2$ and $R'_{ij}\defeq b_i \abs{\braket{a_i}{b_j}}^2$. Note that both $\{R_{ij}\}$ and $\{R'_{ij}\}$ form probability distributions over $[N^2]$. Also note that $\F(R,R') = \Tr(\sqrt{\rho}\sqrt{\sigma})$. Consider (using relation between fidelity and $\ell_1$ distance, Facts~\ref{fact:fidelityvstrace} and~\ref{fact:tracefidelityequi}),
\begin{align}
 \sum_{i,j}\min(R_{ij},R'_{i,j}) & = 1 - \frac{1}{2}\onenorm{R-R'}  \geq 1 - \sqrt{1-\F(R,R')^2} \nonumber \\
 &= 1 - \sqrt{1-(\Tr \sqrt{\rho}\sqrt{\sigma})^2}  \geq  1-\sqrt{\onenorm{\rho-\sigma} - \frac{1}{4}\onenorm{\rho-\sigma}^2} \label{eq:minR}. 
  \end{align}
  Combining Equations~\eqref{eq:thetatau} and~\eqref{eq:minR} we get the desired.
\end{proof}

\begin{claim}
Let $M=\{M_1,M_2\ldots M_w\}$ be a projective measurement  in the support of $A_1A_2$. Let $E=\sum_{i=1}^w M_i\otimes M_i$. Then $\Tr (E \ketbra{\theta}) = 1$.
\end{claim}

\begin{proof}
Since $M_i$ is a projector in the support of $A_1A_2$, we have $(M_i\otimes M_i)\ket{\theta} = (M_i\otimes I)\ket{\theta}$. Hence, \[\bra{\theta}E\ket{\theta}=\sum_i\bra{\theta}M_i\otimes M_i\ket{\theta} = \sum_i\bra{\theta}M_i\otimes I\ket{\theta} = 1 . \myqedhere \]
\end{proof}
Finally  using monotonicity of fidelity under quantum operation (Fact~\ref{fact:monotonequantumoperation}) and Claim~\ref{overlaptau} we get the second part of the theorem as follows. 
\[\sqrt{\Tr (E\tau)} \geq  \F(\tau , \ketbra{\theta}) = \sqrt{\bra{\theta}\tau\ket{\theta}} \geq \left(1-\sqrt{\onenorm{\rho-\sigma} - \frac{1}{4}\onenorm{\rho-\sigma}^2}\right)^{3/2} . \myqedhere \]

\end{proof}

\section{Conclusion and Open questions}
\label{sec:conclusion}
 
We have described two one shot quantum protocols, one of which has been applied to direct sum problem in quantum communication complexity. Our first protocol is a compression protocol, in which communication of a quantum state $\rho$ (held by Alice) can be made much smaller than $\log(|\text{supp}(\rho)|)$, given a description of an another quantum state $\sigma$ with Bob. This protocol is then used to obtain a direct sum result for one round entanglement assisted communication complexity. It may be noted that this application has been superseded by a recent result of Touchette \cite{Touchette:2015} for bounded round entanglement assisted communication complexity models. 

Our second protocol is a quantum generalization of classical correlated sampling. We show that if Alice and Bob are given descriptions of quantum states $\rho$ and $\sigma$, respectively, then they can create a joint state with marginals $\rho$ (on Alice's side) and $\sigma$ (on Bob's side), such that the joint state is correlated. Any measurement done joint by both parties gives highly correlated outcomes, if $\rho$ and $\sigma$ are close to each other in $\ell_1$ distance. 

Some interesting open questions related to this work are as follows.

\begin{enumerate}
\item Can we show a direct product result for all relations in the  one-way entanglement assisted communication model ?
\item Can we show a direct product result for all relations in the bounded-round   entanglement assisted communication model ?
\item Can we find other interesting applications of the protocols appearing in this work ?
\end{enumerate}

\subsection*{Acknowledgment} We thank Mario Berta, Ashwin Nayak, Mark M. Wilde and Andreas Winter for helpful discussions. We also thank anonymous referees for important suggestions for improvement of the manuscript. This work is supported by the Singapore Ministry of Education Academic Research Fund Tier 3 MOE2012-T3-1-009 and also the Core Grants of the Center for Quantum Technologies (CQT), Singapore. Work of A.S. was done while visiting CQT, Singapore. 

\bibliographystyle{alpha}
\bibliography{references}

\end{document}